\documentclass[10pt, conference]{IEEEtran}
\usepackage{amsmath}
\usepackage{amssymb}
\usepackage{amsfonts}
\usepackage{amscd}
\usepackage{mathrsfs}
\usepackage[final]{graphicx}
\usepackage{graphicx}
\usepackage{color}
\usepackage{url}

\newtheorem{theorem}{Theorem}

\newtheorem{definition}{Definition}

\newtheorem{example}{Example}

\title{High-Rate, Distributed Training-Embedded Complex Orthogonal Designs for Relay Networks}
\begin{document}
\author{
\authorblockN{J. Harshan}
\authorblockA{Dept of ECE,\\ Indian Institute of Science \\
Bangalore 560012, India\\
{\em Email:} \url{harshan@ece.iisc.ernet.in}\\
}
\and
\authorblockN{B. Sundar Rajan}
\authorblockA{Dept of ECE, \\Indian Institute of Science \\
Bangalore 560012, India\\
\hspace{0.2cm} {\em Email:} \url{bsrajan@ece.iisc.ernet.in}\\
}
\and
\authorblockN{Are Hj{\o}rungnes}
\authorblockA{UNIK-University Graduate Center \\
University of Oslo, \\NO-2027, Kjeller, Norway\\
{\em Email:} \url{arehj@unik.no} \\
}
}
\maketitle
%
\begin{abstract}
Distributed Space-Time Block Codes (DSTBCs) from Complex Orthogonal Designs (CODs) (both square and non-square CODs other than the Alamouti design) are known to lose their single-symbol ML decodable (SSD) property when used in two-hop wireless relay networks using amplify and forward protocol. For such a network, in this paper, a new class of high rate, training-embedded (TE) SSD DSTBCs are constructed from TE-CODs. The proposed codes include the training symbols in the structure of the code which is shown to be the key point to obtain high rate as well as the SSD property. TE-CODs are shown to offer full-diversity for arbitrary complex constellations. Non-square TE-CODs are shown to provide higher rates (in symbols per channel use) compared to the known SSD DSTBCs for relay networks with number of relays less than $10.$ 
\end{abstract}

\section{Introduction and Preliminaries}
\label{sec1}
\begin{LARGE}D\end{LARGE}istributed space-time coding has been a powerful technique for achieving spatial diversity in wireless networks with single antenna terminals. An excellent introduction to cooperative communications based on distributed space-time coding in two-hop wireless networks can be seen in \cite{LaW, JiH1} and the references within. The technique involves a two phase protocol where, in the first phase, the source broadcasts the information to the relays and in the second phase, the relays linearly process the signals received from the source and forward them to the destination such that the signal at the destination appears as a Space-Time Block Code (STBC). Such STBCs, generated distributively by the relay nodes, are called Distributed Space-Time Block Codes (DSTBCs).

\indent In a co-located Multiple Input Multiple Output (MIMO) channel, an STBC is said to be Single-Symbol Maximum Likelihood (ML) Decodable (SSD) if the ML decoding metric splits as a sum of several terms, with each term being a function of only one of the information symbols \cite{KhR}. Since the work of \cite{LaW, JiH1}, considerable efforts have been made to design SSD DSTBCs. A DSTBC is said to be SSD if the STBC seen by the destination from the set of relays is SSD.

\indent DSTBCs with single-symbol ML decodability was first introduced for cooperative networks in \cite{YiK}. Further, in \cite{HaR}, high-rate, SSD DSTBCs have been proposed wherein the source performs linear precoding of information symbols before transmitting it to all the relays. For the class of codes proposed in \cite{YiK} and \cite{HaR}, the channel model is such that each relay is assumed to know only the statistics of the channel from the source to itself (but not their realizations). In \cite{YiK2} and \cite{SCR}, SSD DSTBCs are proposed for the case where every relay node is assumed to have the perfect knowledge of the phase component of the channel from the source to the relay. An upper bound on the symbol rate for such a set up is shown to be $\frac{1}{2}$ (in complex symbols per channel use in the second phase) which is independent of the number of relays. However, these codes have exponential decoding delay whereas the codes in \cite{YiK} and \cite{HaR} are of minimal delay. Moreover, in the model considered in \cite{YiK2} and \cite{SCR}, training sequences have to be transmitted from the source to the relays since each relay needs to know the phase component of the channel from the source to itself. Therefore, the source needs to use some of the resources such as power and bandwidth for transmitting the training sequences. In \cite{YiK2} and \cite{SCR}, the number of channel uses spent on transmitting training signals are not accounted in computing the rate of the DSTBCs.

\indent For point to point co-located MIMO channels, complex orthogonal designs (CODs) \cite{TJC, Lia}, coordinate interleaved orthogonal designs (CIODs) \cite{KhR} and Clifford unitary weight designs (CUWDs) \cite{KaR} are well known for their SSD property when used to generate STBCs. Note that, with the assumption of the knowledge of the phase component of the source-relay channel at the relays, all CODs can be constructed as DSTBCs \cite{JiJ}. The extensions of CODs such as CIODs and CUWDs can also be distributively constructed. However, CODs (other than the Alamouti design), CIODs and CUWDs (other than that for 4 antennas) do not retain the SSD property. 
 
\indent In this paper, we propose high rate, training embedded SSD DSTBCs. The proposed codes include the training symbols in the structure of the code which is shown to be the key point to obtain high rate as well as the SSD property. On the similar lines of the work in \cite{YiK2}, \cite{SCR}, the relay nodes are assumed to have the knowledge of the phase component of the channel from the source to itself. In this paper, the number of channel uses spent on transmitting training signals from the source to the relays are accounted in computing the rate of the proposed DSTBCs. The main contributions of this paper and the organization can be summarized as follows:
\begin{itemize}
\item We propose a novel method to construct high rate (in symbols per channel use), SSD DSTBCs for two-hop wireless relay networks based on the amplify and forward protocol. The proposed method has an in-built training scheme for the relays to learn the phase components of their backward channels. The in-built training symbols is shown to be the key point to obtain high rate as well as the SSD property (Section \ref{sec2}).
\item When all the zero entries of a COD (square or non-square) are replaced by a constant, the resulting design is called a Training-Embedded-COD (TE-COD). These are shown to generate SSD DSTBCs. This essentially enables all CODs to be usable as SSD DSTBCs with full-diversity for arbitrary complex constellations. Compared to the existing SSD codes of \cite{SCR} (where the number of channel uses spent in sending the training symbols are not included in calculating the rate of the DSTBCs), the class of non-square TE-CODs are shown to provide higher rates for two-hop networks with number of relays less than $10$ (Section \ref{sec3}).  We highlight that the class of non-square TE-CODs provide higher rates than those in \cite{SCR} even though the number of channel uses spent in sending the training symbols are not included in calculating the rate of the schemes in \cite{SCR}.
\item Simulation results for $4$ relays are presented which show that the proposed scheme performs better than the code presented in \cite{SCR} by 0.5 db (Section \ref{sec4}).

\end{itemize}
\indent \textit{Notations:} Throughout the paper, lower case boldface letters and capital boldface letters are used to represent vectors and matrices respectively. For a complex matrix $\textbf{X}$, the matrices $\textbf{X}^*$, $\textbf{X}^T$,  $\textbf{X}^{H}$, $|\textbf{X}|$, $\mbox{Re}~\textbf{X}$ and $\mbox{Im}~\textbf{X}$ denote, respectively, the conjugate, transpose, conjugate transpose, determinant, real part and imaginary part of $\textbf{X}$. The element in the $r_1$-th row and the $r_2$-th column of the matrix $\textbf{X}$ is denoted by $[\textbf{X}]_{r_1,r_2}$. The $ T\times T$ identity matrix and the $T \times T$ zero matrix are respectively denoted by $\textbf{I}_T$ and $\textbf{0}_T$. The magnitude of a complex number $x$, is denoted by $|x|$ and $E \left[x\right]$ is used to denote the expectation of the random variable $x.$ A circularly symmetric complex Gaussian random vector $\textbf{x},$ with mean $\boldsymbol{\mu}$ and covariance matrix $\mathbf{\Gamma}$ is denoted by $\textbf{x} \sim \mathcal{CSCG} \left(\boldsymbol{\mu}, \mathbf{\Gamma} \right) $. The set of all integers, the real numbers and the complex numbers are respectively, denoted by ${\mathbb Z}$, $\mathbb{R}$ and ${\mathbb C}$ and  $\bf{i}$ is used to represent $\sqrt{-1}.$


\section{Training- Embedded Precoded Distributed Space-Time Coding}
\label{sec2}
\subsection{Signal Model}
The wireless network considered as shown in Fig. \ref{model_network} consists of $K + 2$ nodes, each having a single antenna. There is one source node and one destination node. All the other $K$ nodes are relays. We denote the channel from the source node to the $\lambda$-th relay as $h_{\lambda}$ and the channel from the $\lambda$-th relay to the destination node as $g_{\lambda}$ for $\lambda=1,2, \cdots, K$.
\noindent The following assumptions are made in our model:
\begin{itemize}
\item All the nodes are half duplex constrained.
\item Fading coefficients  $h_{\lambda}$ and $g_{\lambda}$ are i.i.d $ \mathcal{CSCG} \left(0,1 \right)$ with a coherence time interval of at least $N$ and $T$ channel uses respectively, where $N$ and $T$ are the number of channel uses in the first phase and the second phase, respectively.
\item All the nodes are synchronized at the symbol level.
\item Relay nodes have the knowledge of only the phase components of the fade coefficients $h_{\lambda}$.
\item Destination knows all the fade coefficients $g_{\lambda}$, $h_{\lambda}$ for $\lambda = 1, 2, \cdots, K$.
\end{itemize}
\begin{figure}[h]
\centering
\includegraphics[width=3in]{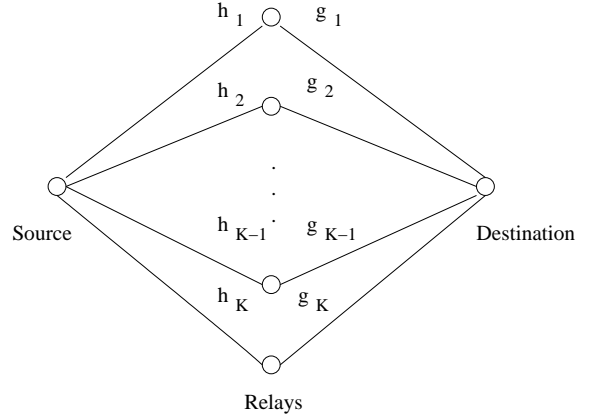}
\caption{Wireless relay network model.}
\label{model_network}
\end{figure}
\begin{figure*} 
\begin{equation} 
\label{concatenation}
\bar{\textbf{r}}_{\lambda} =
e^{-\textbf{i}\angle{h_\lambda}}  \left[e^{\textbf{i}2 (\angle{\alpha}+\angle{h_\lambda})}\textbf{r}_{\lambda}^{*}(1) ~e^{\textbf{i}2 (\angle{\alpha}+\angle{h_\lambda})}\textbf{r}_{\lambda}^{*}(2) ~\cdots ~e^{\textbf{i}2 (\angle{\alpha}+\angle{h_\lambda})}\textbf{r}^{*}_{\lambda}(\lfloor \frac{T-k}{2} \rfloor) ~\textbf{r}^T_{\lambda} \right]^{T} \in \mathbb{C}^{T \times 1}
\end{equation}
\hrule
\end{figure*}
\begin{figure*}
\begin{equation}
\label{code_word}
\textbf{X}  = \left[
\textbf{A}_{1}\bar{\textbf{x}} + \textbf{B}_{1}\bar{\textbf{x}}^{*}  ~~
\textbf{A}_{2}\bar{\textbf{x}} + \textbf{B}_{2}\bar{\textbf{x}}^{*}  ~~
\cdots ~~
\textbf{A}_{K}\bar{\textbf{x}} + \textbf{B}_{K}\bar{\textbf{x}}^{*}
\right] \in \mathbb{C}^{T \times K}
\end{equation}
\end{figure*}
\begin{figure*}
\begin{equation}
\label{covariance_matrix}
\textbf{R} = {\frac{P_{2}T}{P_{r}}}\left[ \sum_{\lambda=1}^{K} |g_{\lambda}|^{2}\left\lbrace \textbf{A}_{\lambda}\textbf{A}_{\lambda}^{H} + \textbf{B}_{\lambda}\textbf{B}_{\lambda}^{H}\right\rbrace \right]  + \textbf{I}_{T} \in \mathbb{C}^{T \times T}\\
\end{equation}
\end{figure*}
\begin{figure*}
\begin{equation}
\label{ML}
\hat{\textbf{x}} = \mbox{arg} \, \min_{\textbf{x} \in \mathcal{S}} \left[ -2 \mbox{Re} \left(\sqrt{\frac{P_{1}P_{2}NT}{P_{r}}}\textbf{g}^{H}\textbf{X}^{H}\textbf{R}^{-1}\textbf{y}\right) + {\frac{P_{1}P_{2}NT}{P_{r}}}\textbf{g}^{H}\textbf{X}^{H}\textbf{R}^{-1}\textbf{X}\textbf{g} \right] \in \mathbb{C}^{k \times 1}
\end{equation}
\hrule
\end{figure*}

\noindent The source is equipped with a codebook $\mathcal{S}$ = $\left\{ \textbf{x}_{1},\, \textbf{x}_{2},\, \textbf{x}_{3},\, \cdots, \textbf{x}_{L} \right\} $ consisting of information vectors $\textbf{x}_{l} \in \mathbb{C}^{N \times 1}$ such that $E\left[\textbf{x}_{l}^{H}\textbf{x}_{l}\right]$ = $1$. The information vectors are of the form,
\begin{equation*}
\textbf{x} = [ \underbrace{\alpha ~\alpha ~\cdots ~\alpha}_{\lceil \frac{T-k}{2} \rceil ~times} ~x_{1} ~x_{2} \cdots ~x_{k}]^T \in \mathbb{C}^{N \times 1} 
\end{equation*}
where the complex variables $x_{1}, x_{2} \cdots x_{k}$ take values from a complex signal set denoted by $\mathcal{M},$  $\alpha \in \mathbb{C}$ is a non-zero complex constant chosen as the training symbol and $N= \lceil \frac{T-k}{2} \rceil+k.$  The value of $\alpha$ is chosen such that the condition $E\left[\textbf{x}_{l}^{H}\textbf{x}_{l}\right]$ = 1 is satisfied. The value of $\alpha$ is assumed to be known to all the relays and the destination. In the first phase, the source broadcasts the vector $\textbf{x}$ to all the $K$ relays (but not to the destination which is assumed to be located far from the source). 

The received vector at the $\lambda$-th relay is given by $\textbf{r}_{\lambda} = \sqrt{P_{1}N}h_{\lambda}\textbf{x} + \textbf{n}_{\lambda} \in \mathbb{C}^{N \times 1}$, for all $\lambda = 1,2,\cdots, K$ where $\textbf{n}_{\lambda} \sim \mathcal{CSCG} \left(\textbf{0}_{N \times 1},\textbf{I}_{N} \right) $ is the additive noise at the $\lambda$-th relay and $P_{1}$ is the total power used at the source node for every channel use. Using the $N = \lceil \frac{T-k}{2} \rceil + k$ length vector, $\textbf{r}_{\lambda}$, the $\lambda$-th relay constructs the $T-$length new vector $\bar{\textbf{r}}_{\lambda}$ given by \eqref{concatenation} shown at the top of this page, where $\textbf{r}_{\lambda}(i)$ denotes the $i$-th component of the vector $\textbf{r}_{\lambda}$. The $\lambda$-th relay is assumed to obtain a perfect estimate of the phase component of $h_{\lambda}$ using the training symbols sent during the first $\lceil \frac{T-k}{2} \rceil$ channel uses in the first phase. This has enabled the phase compensation in \eqref{concatenation} which can also be  given by 
\begin{equation*}
\bar{\textbf{r}}_{\lambda} = \sqrt{P_{1}N}|h_{\lambda}|\bar{\textbf{x}} + \bar{\textbf{n}}_{\lambda} \in \mathbb{C}^{T\times 1}
\end{equation*}
where 
\begin{equation}
\label{concatenation1}
\bar{\textbf{x}} = [ \underbrace{\alpha ~\alpha ~\cdots ~\alpha}_{T-k ~\mbox{times}} ~{x}_{1}~{x}_{2} \cdots {x}_{k}]^{T} \in \mathbb{C}^{T \times 1}.
\end{equation}

Note that the concatenating operation in \eqref{concatenation} continues to keep the components of $\bar{\textbf{n}}_{\lambda}$ identically distributed and uncorrelated to each other. \\
\indent In the second phase, all the relay nodes are scheduled to transmit $T$ length vectors to the destination simultaneously. Each relay is equipped with a fixed pair of matrices $\textbf{A}_{\lambda}$, $\textbf{B}_{\lambda} \in \mathbb{C}^{T \times T}$ and is allowed to linearly process the vector $\bar{\textbf{r}}_{\lambda}$. The $\lambda$-th relay is scheduled to transmit
\begin{equation}
\label{amlify_technique}
\textbf{t}_{\lambda} = \sqrt{\frac{P_{2}T}{P_{r}}}\left\lbrace \textbf{A}_{\lambda}\bar{\textbf{r}}_{\lambda} + \textbf{B}_{\lambda}\bar{\textbf{r}}_{\lambda}^{*}\right\rbrace \in \mathbb{C}^{T \times 1}
\end{equation}
where $P_{2}$ is the total power used at each relay for every channel use in the second phase and $P_{r}$ is the average norm of the vector $\bar{\textbf{r}}_{\lambda}$. The vector received at the destination is given by
\begin{equation*}
\textbf{y} = \sum_{\lambda = 1}^{K} g_{\lambda}\textbf{t}_{\lambda} + \textbf{w} \in \mathbb{C}^{T \times 1}
\end{equation*}
\noindent where $\textbf{w} \sim \mathcal{CSCG} \left(\boldsymbol{0}_{T \times 1},\textbf{I}_{T} \right)$ is the additive noise at the destination. Substituting for $\textbf{t}_{\lambda}$, $\textbf{y}$ can be written as
\begin{equation*}
\textbf{y} = \sqrt{\frac{P_{1}P_{2}NT}{P_{r}}}\textbf{X}\textbf{g}  + \textbf{n}
 \in \mathbb{C}^{T \times 1}
\end{equation*}
\noindent where
\begin{itemize}
\item $\textbf{n} = \sqrt{\frac{P_{2}T}{P_{r}}}\left[ \sum_{\lambda=1}^{K} g_{\lambda}\left\lbrace \textbf{A}_{\lambda} \bar{\textbf{n}}_{\lambda} + \textbf{B}_{\lambda} \bar{\textbf{n}}_{\lambda}^{*}\right\rbrace \right]  + \textbf{w} \in \mathbb{C}^{T \times 1}   \label{bfN}.$
\item The equivalent channel \textbf{g} is given by $[|h_{1}|g_{1} ~ |h_{2}|g_{2} ~ \cdots ~ |h_{K}|g_{K} ]^{T} \in \mathbb{C}^{K \times 1}.$
\item Every codeword $\textbf{X} \in \mathbb{C}^{T \times K}$ which is of the form \eqref{code_word} (shown at the top of this page) is a function of the information vector $\textbf{x}$ through $\bar{\textbf{x}}$.
\end{itemize}

The covariance matrix $\textbf{R} \in \mathbb{C}^{T \times T}$ of the noise vector $\textbf{n}$ is given in \eqref{covariance_matrix} (top of this page). Note that $\textbf{R}$ depends on the choice of the relay matrices $\textbf{A}_{\lambda}$ and $\textbf{B}_{\lambda}.$ The relay matrices needs to be chosen such that the resulting code seen by the destination is SSD. 

The Maximum Likelihood (ML) decoder for $\textbf{x}$ is given by $\hat{\textbf{x}}$ shown in \eqref{ML} at the top of this page.
\begin{definition}
\label{def_pdstbc}
\noindent The collection $\mathcal{C}$ of $T \times K$ codeword matrices given by \eqref{code_word},
\begin{equation}
\label{dstbc}
\mathcal{C} = \left\{ \textbf{X} \mid \forall ~\textbf{x} \in \mathcal{S} \right\}
\end{equation}

\noindent is called a Training-Embedded Distributed Space-Time Block Code (TE-DSTBC) which is determined by the sets $\left\lbrace  \textbf{A}_{\lambda}, \textbf{B}_{\lambda}\right\rbrace$ and $\mathcal{S}$.
\end{definition}

\indent Note that unlike the existing DSTBCs, TE-DSTBCs contain the training symbols in the code structure along with the information symbols justifying their name. In the following section, we show that this training-embedding enables construction of SSD TE-DSTBCs.

\section{TE-DSTBC from TE-CODs}
\label{sec3}

\indent In this section, we construct two classes of TE-DSTBC (square and non-square TE-DSTBC) that are single-symbol ML decodable at the destination. The proposed designs are derived from the well known class of complex orthogonal designs (CODs) \cite{TJC}. The proposed class of complex designs are introduced in the following definition.
\begin{definition}
\label{def_alpha_cod}
Let the $T \times K$ matrix $\textbf{X}$ represent a COD in $k$ complex variables. If the zeros in the design $\textbf{X}$ are replaced by a non-zero constant say $\alpha \in \mathbb{C}$, then we refer $\textbf{X}$ as a TE-COD. 
\end{definition}

\indent Note that the above definition holds both for  the classes of  square CODs (when $T = K$) as well as non-square CODs (when $T > K$).

\begin{example}
\label{example_1}
For the well known $4 \times 4$ COD  \cite{TJC} of rate $\frac{3}{4}$,
with $x_{1}$, $x_{2}$ and $x_{3}$ being the complex variables, the corresponding TE-COD is given by,
\begin{equation}
\textbf{X}_{\mbox{TE-COD}} = \left[\begin{array}{rrrr}
x_{3} & \alpha & x_{2} & x_{1}\\
\alpha & x_{3} & x_{1}^{*} & -x_{2}^{*}\\
x_{2}^{*} & x_{1} & -x_{3}^{*} & \alpha\\
x_{1}^{*} & -x_{2} & \alpha & -x_{3}^{*}\\
\end{array}\right].\\
\end{equation}
\end{example}
In general, given a $T \times K$ TE-COD, $\textbf{X}_{\mbox{TE-COD}}$ in $k$ variables, every column of $\textbf{X}$ contains exactly $k$ distinct variables and $T - k$ copies of $\alpha$. Since $\textbf{X}$ is a linear design \cite{HaH} in the constant $\alpha$ and the variables $x_{i}$'s, the design $\textbf{X}$ can also be written as 

{\small
\begin{equation}
\label{alpha_cod}
\textbf{X}_{\mbox{TE-COD}}  = \left[ \textbf{C}_{1}\bar{\textbf{x}} + \textbf{D}_{1}\bar{\textbf{x}}^{*} ~~  \cdots ~~ \textbf{C}_{K}\bar{\textbf{x}} + \textbf{D}_{K}\bar{\textbf{x}}^{*} \right] \in \mathbb{C}^{T \times K}
\end{equation}
}
\noindent where 
\begin{equation}
\label{x_vector}
\bar{\textbf{x}} = [ \underbrace{\alpha ~\alpha ~\cdots ~\alpha}_{T-k ~\mbox{times}} ~x_{1}~x_{2} \cdots x_{k}]^{T} \in \mathbb{C}^{T \times 1}
\end{equation}
and $\textbf{C}_{\lambda}, \textbf{D}_{\lambda} \in \mathbb{C}^{T \times T},$  $\lambda =1,2,\cdots,K,$ are the column-vector representation matrices of $\textbf{X}_{\mbox{TE-COD}}$ \cite{Lia}. The number of $\alpha$'s in the vector $\bar{\textbf{x}}$ is equal to the number of $\alpha$'s in every column of TE-COD. The following theorem provides an important relation satisfied by the matrices $\textbf{C}_{\lambda}, \textbf{D}_{\lambda}$ of TE-CODs.
\begin{theorem}
\label{thm1}
The column-vector representation matrices $\textbf{C}_{\lambda}, \textbf{D}_{\lambda}$ of a TE-COD, $\textbf{X}_{\mbox{TE-COD}}$ (as represented in \eqref{alpha_cod}) can be chosen to satisfy the following relation,
\begin{equation}
\label{relay_unitary_matrix}
\textbf{C}_{\lambda}\textbf{C}^{H}_{\lambda} + \textbf{D}_{\lambda}\textbf{D}^{H}_{\lambda} = \textbf{I}_{T} ~\forall ~\lambda  = 1 \mbox{ to }K.\\
\end{equation}
\end{theorem}
\begin{proof} 
Consider the column vector representation of TE-CODs as given in \eqref{alpha_cod}. Since the entries of  $\textbf{X}_{\mbox{TE-COD}}$ are of the form $\alpha$, $\pm x_{i}$ and $\pm x^{*}_{i} ~\forall~i = 1$ to $k$ and the vector $\bar{\textbf{x}}$ is given by \eqref{x_vector}, 
it is straightforward to verify that the matrices $\textbf{C}_{i_1}, \textbf{D}_{i_2},$ $i_1,i_2=1,2,\cdots, K,$ satisfy the following three properties:
\begin{itemize}
\item The entries of the matrices $\textbf{C}_{i_1}, \textbf{D}_{i_2}$ are $0, \pm 1$.
\item The matrices $\textbf{C}_{i_1}, \textbf{D}_{i_2}$ can have at most one non-zero entry in every row.
\item The two matrices $\textbf{C}_{i_1}$ and $\textbf{D}_{i_1}$ do not contain non-zero entries in the same row.
\end{itemize}
Note that every complex variable appears exactly once (either as $\pm x_{i}$ or $\pm x^{*}_{i}$) in every column of the design. Without loss of generality, let us assume that $l_{\lambda}$ out of the $k$ complex variables which appear in the $\lambda$-th column of the design, $\lambda=1,2,\cdots,K,$ are of the form $\pm x_{i}$. Then, the matrix $\textbf{C}_{\lambda}$ must have $T-k+l_{\lambda}$ non-zero rows (where $l_{\lambda}$ non-zero rows are for the variables and the remaining non-zero rows are for the $\alpha$'s). Further, as the remaining $k-l_{\lambda}$ variables appear as conjugates (i.e., of the form $\pm x^{*}_{i}$), the matrix $\textbf{D}_{\lambda}$ must have $k-l_{\lambda}$ non-zero rows.

Since there are $T-k$ entries that are $\alpha$ in the vector $\bar{\textbf{x}}$, the non-zero entries in the $T-k$ non-zero rows, which are alloted for the $T-k$ copies of $\alpha$ of $\textbf{C}_{\lambda}$ can be chosen to appear in different columns. Therefore, the columns of $\textbf{C}_{\lambda}$ and $\textbf{D}_{\lambda}$ will have exactly one non-zero entry and hence they satisfy the relations given by \eqref{relay_unitary_matrix}.
\end{proof}
\subsection{Distributed Construction of TE-CODs}
\label{sec3_subsec1}
With reference to the distributed space-time coding technique proposed in Section \ref{sec2}, in this section, we describe how to choose the sets $\left\lbrace  \textbf{A}_{\lambda}, \textbf{B}_{\lambda} ~|~\lambda = 1 \mbox{ to } K \right\rbrace$ and $\mathcal{S}$ such that a $T \times K$ TE-COD, $\textbf{X}_{\mbox{TE-COD}}$ in $k$ variables can be constructed as the TE-DSTBC given in \eqref{dstbc}. Note that every column of $\textbf{X}_{\mbox{TE-COD}}$ contains exactly $k$ distinct variables and $T - k$ copies of $\alpha$. 

After each relay performs the concatenation operation specified in \eqref{concatenation}, the vector $\bar{\textbf{r}}_{\lambda}$ is given by $\bar{\textbf{r}}_{\lambda} = \sqrt{P_{1}N}|h_{\lambda}|\bar{\textbf{x}} + \bar{\textbf{n}}_{\lambda}$ where the vector $\bar{\textbf{x}}$ is given by \eqref{concatenation1}. Hence, the column vector representation matrices $\textbf{C}_{\lambda}$ and $\textbf{D}_{\lambda}$ of $\textbf{X}_{\mbox{TE-COD}}$ are the same as the relay matrices, $\textbf{A}_{\lambda}$ and  $\textbf{B}_{\lambda}$ respectively. With the above choice on the sets $\left\lbrace \textbf{A}_{\lambda}, \textbf{B}_{\lambda} ~|~\lambda = 1 \mbox{ to } K \right\rbrace$ and $\mathcal{S}$, a $T \times K$ TE-COD, $\textbf{X}_{\mbox{TE-COD}}$ in $k$ variables can be constructed as a TE-DSTBC.
\begin{example}
\label{example_2}
To construct the TE-COD given in Example \ref{example_1}, the following ingredients are required at the various terminals. We have $T=K=4$ and $k=3.$ The set $\mathcal{S}$ is given by $\mathcal{S} = \{ [\alpha ~x_{1} ~x_{2} ~x_{3}]^T ~|~ \forall ~ x_{i} \in \mathcal{M} \}.$ The corresponding relay matrices $\textbf{A}_{\lambda}, \textbf{B}_{\lambda}$ are given by

{\footnotesize
\begin{equation*}
\textbf{A}_{1} = \left[\begin{array}{rrrr}
0 & 0 & 0 & 1\\
1  & 0 & 0 & 0\\
0 & 0 & 0 & 0\\
0 & 0 & 0 & 0\\
\end{array}\right];~ 
\textbf{B}_{1} = \left[\begin{array}{rrrr}
0 & 0 & 0 & 0\\
0  & 0 & 0 & 0\\
0 & 0 & 1 & 0\\
0 & 1 & 0 & 0\\
\end{array}\right];~
\end{equation*}
\begin{equation*}
\textbf{A}_{2} = \left[\begin{array}{rrrr}
1 & 0 & 0 & 0\\
0  & 0 & 0 & 1\\
0 & 1 & 0 & 0\\
0 & 0 & -1 & 0\\
\end{array}\right];~ 
\textbf{B}_{2} = \textbf{0}_{4};
\end{equation*}
\begin{equation*}
\textbf{A}_{3} = \left[\begin{array}{rrrr}
0 & 0 & 1 & 0\\
0  & 0 & 0 & 0\\
0 & 0 & 0 & 0\\
1 & 0 & 0 & 0\\
\end{array}\right] ;~ \textbf{B}_{3} = \left[\begin{array}{rrrr}
0 & 0 & 0 & 0\\
0  & 1 & 0 & 0\\
0 & 0 & 0 & -1\\
0 & 0 & 0 & 0\\
\end{array}\right];~
\end{equation*}
\begin{equation*}
\textbf{A}_{4} = \left[\begin{array}{rrrr}
0 & 1 & 0 & 0\\
0  & 0 & 0 & 0\\
1 & 0 & 0 & 0\\
0 & 0 & 0 & 0\\
\end{array}\right];~\textbf{B}_{4} = \left[\begin{array}{rrrr}
0 & 0 & 0 & 0\\
0  & 0 & -1 & 0\\
0 & 0 & 0 & 0\\
0 & 0 & 0 & -1\\
\end{array}\right].
\end{equation*}
}
To implement the above design, the number of channel uses required in the first phase is $4$ (3 channel uses for the variables and the rest for transmitting $\alpha$). The number of channel uses in the second phase is also $4$. Hence, the rate of this scheme is $\frac{3}{8}$.
\end{example}
\subsection{On the Single-Symbol ML Decodable Property of Distributed TE-CODs}
\label{sec3_subsec2}
Note that excluding the scaling factors and the constant terms, the ML decoding metric given in \eqref{ML} is a function of the following two terms (i) $\textbf{g}^{H}\textbf{X}^{H}\textbf{R}^{-1}\textbf{y}$ and (ii) $\textbf{g}^{H}\textbf{X}^{H}\textbf{R}^{-1}\textbf{X}\textbf{g}$ where $\textbf{R}$ is a function of the set, $\left\lbrace \textbf{A}_{\lambda}, \textbf{B}_{\lambda}\right\rbrace$ as given in \eqref{covariance_matrix}. Also, from the results of Theorem \ref{thm1}, for the class of TE-CODs,  $\textbf{R}$ is a scaled identity matrix and hence the matrix $\textbf{X}^{H}\textbf{R}^{-1}\textbf{X}$ is the same as $\textbf{R}^{-1}\textbf{X}^{H}\textbf{X}$. Since $\textbf{X}$ is an TE-COD, the matrix $\textbf{X}^{H}\textbf{X}$ can be written as a sum of $k$ matrices where each matrix is strictly a function of only one of the real variables. For example, the matrix $\textbf{X}^{H}\textbf{X}$ for the square TE-COD for $4$ relays given in Example \ref{example_2}, is given by \eqref{XHX_COD_square_example} (at the top of the next page) in which, since $\textbf{X}^{H}\textbf{X}$ is a Hermitian matrix, we only present the elements on and above the main diagonal elements of $\textbf{X}^{H}\textbf{X}$.
\begin{figure*}
\begin{equation}
\label{XHX_COD_square_example}
\textbf{X}^{H}\textbf{X} = \left[\begin{array}{cccccccc}
|\alpha|^{2} + \sum_{i = 1}^{3} |x_{i}|^{2} & 2\mbox{Re}(x_{3}^{*}\alpha) & 2\mbox{Re}(x_{1}^{*}\alpha^{*}) & 2{\bf i}\mbox{Im}(x_{2}\alpha)\\
* & |\alpha|^{2} + \sum_{i = 1}^{3} |x_{i}|^{2} & 2{\bf i}\mbox{Im}(x_{2}\alpha^{*}) & 2\mbox{Re}(x_{1}\alpha^{*})\\
* & * & |\alpha|^{2} + \sum_{i = 1}^{3} |x_{i}|^{2} & -2\mbox{Re}(x_{3}\alpha)\\* & * & * & |\alpha|^{2} + \sum_{i = 1}^{3} |x_{i}|^{2}\\
\end{array}\right].
\end{equation}
\hrule
\end{figure*}
Note that, $\textbf{X}^{H}\textbf{X}$ is not diagonal since all the $0$'s have been replaced by $\alpha$. Hence, the ML decoding metric splits as a sum of several terms, with each term being a function of only one of the variables. Thus, when TE-CODs are applied as TE-DSTBC, every variable can be decoded independent of the other complex variables. Notice that when $\alpha=0,$ the matrix $\textbf{X}^{H}\textbf{X}$ in \eqref{XHX_COD_square_example} becomes a scaled identity matrix corresponding to the well known CODs.
\subsection{Full Diversity of Distributed TE-CODs}
\label{sec3_subsec3}
From the results of \cite{JiH1}, a TE-DSTBC is fully diverse if for any two distinct codewords $\textbf{X}_{1}$ and $\textbf{X}_{2}$ of a TE-DSTBC, the matrix $(\textbf{X}_{1} - \textbf{X}_{2})^{H} (\textbf{X}_{1} - \textbf{X}_{2})$ is full rank. Since we employ a TE-COD to generate the TE-DSTBC, the difference matrix $\textbf{X}_{1}-\textbf{X}_{2}$ gets a $0$ at the position where there is $\alpha$ in the design and hence the matrix $(\textbf{X}_{1} - \textbf{X}_{2})^{H} (\textbf{X}_{1} - \textbf{X}_{2})$ will be a diagonal one with full rank. Thus, TE-DSTBC generated from TE-CODs have full diversity property for arbitrary signal sets.
\subsection{Rate of Distributed TE-CODs in Symbols per Channel Use}
\label{sec3_subsec4}
In our proposed scheme, the total number of channel uses involving both the first phase and the second phase is $\lceil \frac{T-k}{2}\rceil + k + T$. Therefore the rate of our scheme in symbols per channel use is $\frac{k}{\lceil \frac{T-k}{2}\rceil + k + T}$ for both the square and the non-square TE-CODs. The rate for square TE-CODs with $2^a$ number of relays is $
 \frac{a + 1}{a + 1 + \lceil \frac{2^{a} - a - 1}{2} \rceil + 2^{a}}$
complex symbols per channel use and the rate for non-square TE-CODs, from maximal rate CODs, for the cases of $2m$ or $2m-1$ relays is easily calculated to be
$ \frac{\frac{T(m+1)}{2m}}{\frac{T(m+1)}{2m} + T + \lceil T(\frac{1}{4} - \frac{1}{4m})\rceil}
$
complex symbols per channel use.

In \cite{SCR}, SSD-DSTBCs have been constructed with rate $\frac{1}{2}$ in complex symbols per channel use in the second phase. If the number of channel uses in the first phase is also considered, then the rate of such codes will be $\frac{1}{3}$ (which doesn't include the number of channel uses needed for training in the first phase to calculate the rate). In Table \ref{rate_table_designs_1} (shown at the top of the next page), we list the rates (including the channel uses in both the phases) of TE-DSTBCs from non-square TE-CODs (which includes the number of channel uses for training in the first phase to calculate the rate) for different values of $K$. When compared with the codes in \cite{SCR}, (which doesn't include the channel uses needed for training in the first phase) it is clear that for networks with $K < 10$, TE-DSTBCs from non-square TE-CODs provide higher rate than those of the codes in \cite{SCR}. In \cite{SCR}, if the number of channel uses spent on transmitting the training symbols (from the source to the relays) is also included in calculating the rate, then the rate of such DSTBCs will be lesser than $\frac{1}{3}$ and hence non-square TE-CODs provide higher rate gains than those listed in Table \ref{rate_table_designs_1}.   
\begin{table*}
\caption{Overall Rates (First and Second Phase) of TE-DSTBC from Non-square TE-CODs}
\begin{center}
\begin{tabular}{|c|c|c|c|c|c|c|c|c|c|c|c|c|c|c|c|}
\hline  & $K = 4$ & $K = 5$ & $K = 6$ & $K = 7$ & $K = 8$ & $K = 9$ & $K = 10$ \\
\hline SSDs from non-square TE-CODs & $\frac{3}{8}> \frac{1}{3}$ & $\frac{5}{14}> \frac{1}{3}$ & $\frac{4}{11}> \frac{1}{3}$ & $\frac{35}{102}> \frac{1}{3}$ & $\frac{70}{203} > \frac{1}{3}$ & $\frac{126}{378}=\frac{1}{3}$ & $\frac{251}{756} < \frac{1}{3}$ \\
\hline 
\end{tabular} 
\end{center}
\label{rate_table_designs_1}
\end{table*}
\subsection{Channel estimation at the destination}
\label{sec3_subsec5}
For TE-DSTBCs, we assume one of the following two methods of channel estimation at the destination. 
\begin{itemize}
\item Since the destination receives a linear combination of the information symbols and the training symbols, it can possibly estimate all the channel gains using the symbols received during the $T$ channel uses in the second phase. For a background on channel estimation with superposition pilot sequence, we refer the readers to \cite{HoT}. Note that for the DSTBCs proposed in \cite{SCR}, separate training symbols are needed (from the relays to the destination) for the destination to estimate the channels. As a result, the proposed scheme provides further advantage in the overall rate (when the number of channel uses in sending training symbols from the relays to the destination is also included in the calculation of the rate) compared the schemes in \cite{SCR}. 
\item For TE-DSTBCs, additional training symbols can be transmitted from the relays to the destination for channel estimation (this is apart from the training symbols transmitted along with the information symbols). Since, additional training symbols are also needed for the codes proposed in \cite{SCR}, the existing rate advantage of our scheme over the scheme in \cite{SCR} still holds.
\end{itemize}
\section{Simulations Results}
\label{sec4}
In this section, we provide the performance comparison (in terms of the bit error rate) between the DSTBC from TE-CODs (given in Example \ref{example_1}) and the DSTBC proposed in \cite{SCR} 
 for $K = 4$. Note that both the codes have single-real-symbol ML decodable property for QAM signal sets. Throughout this section, the designs used in \cite{SCR} are referred as "CODs from RODs". For $K = 4$, the rates (in complex symbols per channel use in the second phase) of the DSTBCs from TE-COD and "COD from ROD" are respectively $\frac{3}{4}$ and $\frac{1}{2}.$ For the two codes, both the number of channel uses and the energy consumption in the first phase are the same. However, the number of channel uses and the energy consumption in the second phase are different for the two codes. Hence, for a fair comparison, we make the bits per channel use (bpcu) in the second phase same for both the codes, in particular, we make it equal to $1.5$ bpcu for the simulation purpose. To achieve the common rate of $1.5$ bpcu in the second phase, the TE-COD and the "COD from ROD" respectively employs 4-QAM signal set $\{-1+{\bf i}, 1+{\bf i},-1-{\bf i}, 1-{\bf i}\}$ and $8$-QAM signal set $\{-3+{\bf i}, -1+{\bf i}, 1+{\bf i},   3+{\bf i}, -3-{\bf i}, -1-{\bf i}, 1-{\bf i}, 3+{\bf i}\}$ to construct the DSTBCs. Note that the $8$-QAM signal set is not energy efficient; a more energy efficient $8$-point QAM is $\{ -1+3{\bf i}, 3+3{\bf i}, -3+{\bf i},   1+{\bf i}, -1-{\bf i}, 3-{\bf i}, -3-3{\bf i}, 1-3{\bf i}\}.$  However, with the use of the more energy efficient $8$-point QAM, real symbol ML decodable property will be lost (the ML decoder in such a case will be single-complex symbol decodable). Hence, we use the $8$-QAM constellation in our simulations. The BER performance of both codes are plotted against energy used per bit  in Fig. \ref{plot} which shows that TE-COD performs better than "COD from ROD" by 0.5 db.\\
\begin{figure}[h]
\centering
\includegraphics[width=3.5in]{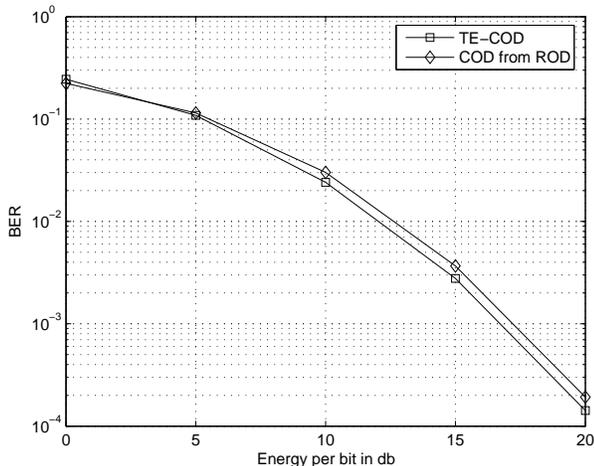}
\vspace*{-1.0cm}
\caption{BER comparison between the DSTBC from TE-COD and the DSTBC from "CODs from RODs" for $K = 4$ with $1.5$ bpcu}
\label{plot}
\vspace*{-0.5cm}
\end{figure}
\section{Discussion and Conclusions}
\label{sec5}
In this paper, through a training based distributed space-time coding technique, we have shown to construct the variants of the well known class of CODs in two-hop relay networks using amplify and forward protocol. The inclusion of training symbols in to the structure of the code has been shown to provide high rate as well as the SSD property for the constructed codes. This idea can be extended to construct DSTBCs from other SSDs like CIODs and CUWDs existing for point to point co-located MIMO channels to two-hop wireless networks \cite{HRH}.

\end{document}